\documentclass[12pt]{amsart}
\usepackage{amsfonts,amssymb,amsmath,amsthm}
\usepackage{url}
\usepackage{enumerate}

\newtheorem{theorem}{Theorem}[section]
\newtheorem{lemma}[theorem]{Lemma}
\newtheorem{proposition}[theorem]{Proposition}
\newtheorem{corollary}[theorem]{Corollary}
\theoremstyle{definition}

\newtheorem{remark}[theorem]{Remark}
\numberwithin{equation}{section}

\newcommand{\R}{\mathbb{R}}
\newcommand{\N}{\mathbb{N}}

\newcommand{\tr}{\mathrm{tr }}

\author{V. R. Bazao}   
\address{
Faculdade de Ci\^encias Exatas e Tecnologias, UFGD, Dourados, MS, 79804-970  Brazil}

\author{T. O. Carvalho}   
\address{Departamento de Matem\'atica, UEL, CP 10011, Londrina, PR, 86057-970
Brazil}

\author{C. R. de Oliveira}   
\address{Departamento de Matem\'atica, UFSCar, S\~ao Carlos, SP, 13560-970 Brazil}


\begin{document}

\title[Fractal dimensions of the period doubling operator]{Lower bounds for fractal dimensions of spectral measures of the period doubling  Schr\"odinger operator}

\begin{abstract}
It is shown that there exits a lower bound $\alpha>0$ to the Hausdorff dimension of the spectral measures of the one-dimensional period doubling substitution Schr\"odinger operator, and, generically in the hull of such sequence, $\alpha$ is also a lower bound to the upper packing dimension of spectral measures. 
\end{abstract}
\maketitle

\section{Main results}\label{intro}
One-dimensional discrete Schr\"odinger operators, with potentials taking a finite number of values along (nonperiodic) almost periodic sequences, have a tendency to present  Cantor spectra of zero Lebesgue measure and purely singular continuous; this is the case of Sturmian potentials~\cite{DKL}, many primitive substitutions~\cite{BG1} as the Fibonacci one~\cite{suto1,suto2} (which, in fact, is a particular Sturmian sequence), Thue-Morse~\cite{belliss}, period doubling~\cite{BBG,Dam2001}, palindromic sequences~\cite{HKS}, etc., and also some nonprimitive substitutions~\cite{deOL,LdeO}. A natural question is about fractal properties of such spectra, with prominent roles played by Hausdorff and packing measures on the line. Usually such fractal properties are not easy to be proven, but nontrivial Hausdorff continuity properties have been obtained for Sturmian potentials whose rotation numbers are of bounded density~\cite{DKL} (see also~\cite{JL2}), and recently this property has been found to be stable under some perturbations when a singular continuous component is persistent~\cite{BCdeO1}.  Note that $\alpha$-Hausdorff continuity implies $\alpha$-packing continuity of a measure~\cite{Fal,Mat,GuarSB1999}, but the reverse does not hold in general.

In this work we give a contribution to such fractal properties; we will state results in terms of the Hausdorff and packing dimensions of spectral measures, so we recall such concepts \cite{Fal,Mat}. Let $ \dim_{\mathrm H}(S)$ and $ \dim_{\mathrm P}(S)$  denote the  Hausdorff and packing  dimensions, respectively, of the set~$S\subset\mathbb R$ and $0< a \le 1$;  if~K indicates~H or~P (for Hausdorff or packing), the {\em upper K dimension} of the finite Borel measure~$\mu$, on~$\mathbb R$, is defined as 
\[
\dim_{\mathrm K}^+(\mu)= \inf \{ \dim_{\mathrm K}(S): \mu(\mathbb R\setminus S)=0, \,S\; \mathrm{a\; Borel\; subset\; of}\; \mathbb R \},
\]
and its {\em lower  K dimension}   as
\[
\dim_{\mathrm K}^-(\mu)= \sup\{a: \mu(S)=0\;\;\mathrm{if}\;\dim_{\mathrm K}(S)<a,\;S\; \mathrm{a\; Borel\; subset\; of}\; \mathbb R\}.
\]
If $S \subset\mathbb{R}$, it is known~\cite{Mat,GuarSB1999} that $\dim_{\mathrm H}(S)\le \dim_{\mathrm P}(S)$, $\dim_{\mathrm H}^-(\mu)\le\dim_{\mathrm P}^-(\mu)$ and $\dim_{\mathrm H}^+(\mu)\le\dim_{\mathrm P}^+(\mu)$; clearly, $\dim_{\mathrm K}^-(\mu)\le\dim_{\mathrm K}^+(\mu)$.

We shall prove lower bounds for lower Hausdorff dimension of the spectral measures of the whole-line operator with potential built along the period doubling substitution sequence and, generically in the hull of this sequence, a lower bound for their upper packing dimensions. Recall that the period doubling substitution~$\xi$ is defined on an alphabet of two letters $\{a,b\}$ as 
\[
\xi \left(a\right)=ab,\qquad \xi \left(b\right)=aa,
\] and with the usual extension by concatenation 
\begin{eqnarray*}
\xi ^{n}\left(a\right)=\xi ^{n-1}\left(ab\right)&=&\xi ^{n-1}
\left(a\right)\xi ^{n-1}\left(b\right),\\ 
\xi ^{n}\left(b\right)=\xi ^{n-1}\left(aa\right)&=&\xi ^{n-1}
\left(a\right)\xi ^{n-1}\left(a\right).
\end{eqnarray*}
The two-sided period doubling substitution sequence~$\varpi$  is obtained as the limit
\begin{equation}\label{potencial2blocos}
 \varpi :=\lim_{n\to\infty}\xi^{2n}(a)\cdot\xi^{2n}(a)=...abaa\cdot abaa...
\end{equation}
 The dot indicates the zeroth position on the right, and~$\varpi$ is a fixed 
point of the substitution, that is, $\xi(\varpi)=\varpi.$ The hull of~$\varpi$ 
is the associated subshift~$\Omega_\varpi$ (the dynamics given by the shift 
operator) of the set of all two-sided sequences~$\omega$ such that all  finite 
subblocks occurring in~$\omega$ also occur in~$\varpi$, and it is a complete 
metric space with the metric of pointwise convergence~\cite{queff}. 

To each sequence $\omega\in \Omega_\varpi$ one associates a potential~$V_\omega$, which is obtained by thinking of $a,b$ as two different real numbers. For instance, $V_\varpi(b)=1$ and $V_\varpi(a)=-4$. 

We are interested in  the  Schr\"odinger operators 
\[
H_{\omega}=\Delta + V_{\omega},\quad \omega\in \Omega_\varpi,
\] acting on $l^2(\mathbb Z)$, with~$\Delta$ denoting the usual discrete Laplacian (i.e., $(\Delta u)(n)= u(n+1)+u(n-1)$). Given a nonzero $\phi\in l^2(\mathbb Z)$, denote the spectral measure of the pair $(H_\omega,\phi)$ by  $\mu_\phi^\omega$. From the spectral viewpoint, it is known~\cite{Dam2001} that $H_\omega$ is purely singular continuous for all $\omega\in\Omega_{\varpi}$, and our main results are the following.

\begin{theorem}\label{teorAlphaPC}
There exists an $\alpha\in (0, 1]$  so that, for all nonzero $\phi\in l^2(\mathbb Z)$: 
\begin{itemize}
\item[(i)] $\dim_{\mathrm P}^-(\mu_\phi^\varpi)\ge\dim_{\mathrm H}^-(\mu_\phi^\varpi)\ge\alpha$;
\item[(ii)]  there is a generic (i.e., dense $G_\delta$) set $\mathcal G\subset \Omega_\varpi$, with $\varpi\in\mathcal G$,  so that $\dim_{\mathrm P}^+(\mu_\phi^\omega)\ge\alpha$, for all $\omega\in\mathcal G$.
\end{itemize} 
\end{theorem}

  For  one-dimensional Schr\"odinger operators with  singular continuous spectrum, this adds to a small list of results on $\alpha$-continuity  outside the scope of (quasi) Sturmian~\cite{DKL,DLquasi}  and sparse potentials \cite{JL2,Zla,Tcherem, CMW,CM,SO}; see also~\cite{KLS,BCdO1} (note that most of these results were obtained around 20 years ago).  We underline that we also know explicitly elements of the generic set~$\mathcal G$.

It is well known that Theorem~\ref{teorAlphaPC} has some dynamical consequences, which we now recall. Let $\delta_{j}$ be the vector that $\delta_j(n)$ takes~1 at $n=j$ and zero otherwise (an element of the canonical basis of~$l^2(\mathbb Z)$). Given a self-adjoint operator~$T$ on~$l^2(\mathbb Z)$, for~$p>0$, let 
\[
\langle X^p_T\rangle(t):=\frac2t\int_0^\infty \sum_{n}|n|^p e^{-2s/t}\;|\langle e^{-isT}\delta_0,\delta_n\rangle|^2\,\mathrm d s
\] denote the average moment of order~$p$ associated with the initial state~$\delta_0$ at time~$t>0$; the lower and upper dynamical exponents~$\beta^-_T(p)$ and $\beta^+_T(p)$ are defined, respectively, as
 \begin{equation}\label{ENo}
\beta^-_T(p):=\liminf_{t\rightarrow\infty}\frac{\ln \langle X^p_T\rangle(t)}{p\ln t}, \quad \beta^+_T(p):=\limsup_{t\rightarrow\infty}\frac{\ln \langle X^p_T\rangle(t)}{p\ln t},
\end{equation}
 and they quantify the transport on the lattice~$\mathbb Z$. For the spectral measure $\mu^T$ of~$(T,\delta_0)$, one has  \cite{guarneri1993, last1996, BCM1997}, for all~$p>0$, 
\[
\beta^-_T(p) \ge \dim_{\mathrm H}^+(\mu^T)   
\]
and~\cite{GuarSB1999}
\[
\beta^+_T(p) \ge \dim_{\mathrm P}^+(\mu^T) .
\] It then follows, by Theorem~\ref{teorAlphaPC} and the above remarks:

\begin{theorem}\label{teorBeta+Generico}
Let $\alpha$ and $\mathcal G$ be as in Theorem~\ref{teorAlphaPC}. Then, for all $p>0$:
\begin{itemize}
\item[(i)] $\beta^-_{H_\varpi}(p) \ge \alpha$; 
\item[(ii)] $\beta^+_{H_\omega}(p) \ge \alpha$ for all $\omega\in\mathcal G$.
\end{itemize}
\end{theorem}

In Section~\ref{sectProofMain}  the proof of our  spectral results are presented; it is based on an auxiliary proposition. Section~\ref{sectProofEstimativas} is devoted to the proof of such proposition, which is the main technical contribution of this work.

\section{Proof of Theorem~\ref{teorAlphaPC}}\label{sectProofMain}
Denote by $\sigma_\omega$ the spectrum of the self-adjoint operator~$H_\omega$. If $u$ is a solution to the eigenvalue equation (note the fixed choice of the potential~$\varpi$)
\begin{equation}\label{eqAutoval}
(H_{\varpi}-E)u=0
\end{equation} with normalized initial conditions (NIC), i.e., $|u(0)|^2+|u(1)|^2=1$, denote by $\|u\|_L$  the truncated norm at $0<L\in \R$ ($[L]$ is the integral part of~$L$), that is,
\[
 \|u\|_L:=\left[\sum_{n=1}^{[L]}|u(n)|^2+ (L-[L])|u([L]+1)|^2\right]^{\frac{1}{2}}.
\]

\begin{proposition}
\label{propEstimativas}
There exist  numbers $0<\gamma_1\le\gamma_2$  so that, for each~$E\in\sigma_\varpi$ and all $L>0$ sufficiently large, all  solutions~$u$ to~\eqref{eqAutoval} with NIC satisfy
\begin{equation}\label{estimativa4}
\frac1{\sqrt{2}}\,L^{\gamma_1}\leq \|u\|_{L}\leq L^{\gamma_2}.
\end{equation}
\end{proposition}

The proof of Proposition~\ref{propEstimativas} is the subject of Section~\ref{sectProofEstimativas}. 

 By invoking Theorem~1 in~\cite{DKL},  Proposition~\ref{propEstimativas} directly implies that the operator~$H_\varpi$ has purely $\alpha$-Hausdorff continuous spectrum with 
 \[
0< \alpha=\frac{2\gamma_1}{\gamma_1+\gamma_2}\le1,
 \] that is,  for  any $\phi\in l^2(\mathbb Z)$,  the spectral measure $\mu_\phi^\varpi$ does not give weight to sets of zero $\alpha$-Hausdorff measure $h^\alpha$. Hence, if $\dim_{\mathrm H}(S)<\alpha$, then $h^\alpha(S)=0$ and so $\mu_\phi^\varpi(S)=0$, which implies that   $\dim_{\mathrm H}^-(\mu_\phi^\varpi)\ge\alpha$. This concludes  Theorem~\ref{teorAlphaPC}(i).

\begin{lemma}\label{lemmaPackUpperGd}
For each $a\in(0,1]$ and  $\phi\in l^2(\mathbb Z)$, the set 
\[
C_{a\mathrm{uPd}}^\phi:=\big\{\omega\in \Omega_\varpi: \dim_{\mathrm P}^+( \mu_\phi^\omega)\ge a  \big\}
\]
is a $G_\delta$ set in~$\Omega_\varpi$.
\end{lemma}
\begin{proof}
This  is a direct consequence of  results in~\cite{CdeOFORUM} for more general (regular) metric spaces of self-adjoint operators.  Theorem~4.2 in~\cite{CdeOFORUM} shows the result for $a=1$, but its proof also applies (a simplification, in fact) for $0<a<1$.
\end{proof}

By Theorem~\ref{teorAlphaPC}(i), for each nonzero vector~$\phi$ one has $\dim_{\mathrm P}^+(\mu_\phi^\varpi)\ge\dim_{\mathrm H}^-(\mu_\phi^\varpi)\ge\alpha$. Since the period doubling sequence is a primitive substitution, the translates (through the shift operator) of~$\varpi$ form a dense set in~$\Omega_\varpi$~\cite{queff}, but each translate  has the same spectral properties as~$H_\varpi$; thus, $C_{\alpha\mathrm{uPd}}^\phi$ is dense in~$\Omega_\varpi$. By Lemma~\ref{lemmaPackUpperGd}, $C_{\alpha\mathrm{uPd}}^\phi$ is also a $G_\delta$ set. This completes the proof of Theorem~\ref{teorAlphaPC}.

\begin{remark}
To the best knowledge of the present authors, there exists no proof that any of the sets  $\{\omega\in \Omega_\varpi: \dim_{\mathrm H}^\pm( \mu_\phi^\omega)\ge a  \}$ is a~$G_\delta$ set; so the restriction of Theorem~\ref{teorAlphaPC}(i) to~$H_\varpi$ (and, if one likes, with the inclusion of  the translates of~$\varpi$, of course).
\end{remark}

\section{Proof of Proposition~\ref{propEstimativas}} \label{sectProofEstimativas}
We begin with some remarks about the period doubling substitution and corresponding transfer matrices. If $u$ is a solution to~\eqref{eqAutoval}, introduce the column vectors
\[U(m+1)=\left( \begin{array}{c}
u(m+1) \\
u(m)\\
\end{array}\right),\]
so that 
\[U(m+1)=M(E,V_{\varpi}(0)\ldots V_{\varpi}(m-1))\,U(1),
\] where $M(m):=M(E,V_{\varpi}(0)\ldots V_{\varpi}(m-1))$ is the 
transfer matrix in this context. Observe that 
$\frac{1}{2}\|U\|_{L}^{2}\leq \|u\|_{L}^{2}\leq\|U\|_{L}^{2}$.

Denote $a_{n}=\xi ^{n}\left(a\right)$ and $b_{n}=\xi ^{n}\left(b\right)$, so that
\begin{equation} \label{eq:relacoes}
a_{n}= a_{n-1}b_{n-1}\quad  \textrm{and}\quad b_{n}=a_{n-1}a_{n-1},
\end{equation}
and both $a_n$ and~$b_n$ have length~$2^{n}$.

For fixed~$E$ and all $n\geq 0$,  write 
\begin{equation}
M_{n}=M_{E}(a_{n})=M(E,V_{\varpi}(0)\ldots V_{\varpi}(2^{n}-1)),
\label{eme_n}
\end{equation}
in particular 
\[
M_0=\begin{bmatrix} 
             E-V(a) & -1 \\ 1 & 0 
             \end{bmatrix}.
             \]
If~$m=\sum_{i=1}^{k}2^{n_{i}}$, then 
\[
M(m)=\prod_{i=0}^{k}M_{n_{k-i}}.
\]

Now we recall well-known  properties of this substitution. Let $x_{n}=x_{n}(E):=\tr\left(M_{E}\left(a_{n}\right)\right)$ and $y_{n} =y_{n}(E):=\tr(M_{E}(b_{n}))$.

\begin{lemma} \label{prop2.1}
For each positive integer~$n$, the words $a_{n}$ and $b_{n}$ spell the same, except  for the rightmost letter. 
\end{lemma}

For each positive integer~$n$, the sequence~$\varpi$ may be uniquely decomposed in blocks $a_n$ and $b_n$, and such decomposition is  called the $n$-partition of~$\varpi$.

\begin{lemma}\label{particaohull}
In the $n$-partition of~$\varpi$, the $b_n$-blocks are isolated and between consecutive occurrences of $b_n$-blocks there occur either one or three $a_n$-blocks. 
\end{lemma}

The following recent result will be crucial to conclude Proposition~\ref{propEstimativas}.

\begin{theorem}[Theorem 1.1 in~\cite{tulimNonl}]\label{thmBoundTrace}
There exists $C\geq 2$ such that, for each $E\in\sigma_\varpi$, 
\[ \limsup_{n\to \infty}\left|x_{n}(E)\right|\leq C \ .\] 
\end{theorem}

\subsection{Lower bound}\label{subsectLowerB} We first derive the lower bound $ 1/\sqrt{2}\, L^{\gamma_{1}}\le \|u\|_{L}$ for  $E\in\sigma_\varpi$ and any solution to~\eqref{eqAutoval} with NIC (and large~$L>0$). The actual value of the constant~$C$ may vary from one equation to another.

\begin{lemma}\label{lema12}
Suppose that there are $C>0$ and a subsequence $n_{k} \to\infty$ so that
\begin{equation}\label{2blocos}
\left\{\begin{array}{l}
V(m)=V(m+n_{k}), \ \ 1\leq m \leq n_k \\
\big|\tr [M_{E}(V_{\omega}(j)\ldots V_{\omega}(j+n_k-1))]\big|\leq C
\end{array}\right.
\end{equation}
Then, for $E\in\sigma_{\varpi}$ and for any $1\leq l \leq n_{k}$, every solution~$u$ to~\eqref{eqAutoval} satisfies
\[\|U\|_{l+2n_{k}}\geq D\|U\|_l\] 
with $D=\left(1+\frac{1}{4C^2}\right)^{\frac{1}{2}}$ and $C$ as in Theorem~\ref{thmBoundTrace}.
\end{lemma}

\begin{proof} Pick  $j\in\{1,\ldots,l\}$. By definition,
\begin{eqnarray*}&U(j+n_{k})=M_{E}(V_{\omega}(j)\ldots V_{\omega}(j+n_k-1)) \,U(j), \\
&U(j+2n_k)=M_{E}(V_{\omega}(j)\ldots V_{\omega}(j+2n_k-1))\,U(j).
\end{eqnarray*}
For $m=j+n_{k}$, since $V(m)=V(m+n_{k})$, it follows that
\[U(j+2n_k)=[M_{E}(V_{\omega}(j)\ldots V_{\omega}(j+n_k-1))]^2\,U(j).\]
Now, by Cayley-Hamilton Theorem, 
\[
U(j+2n_k)-\tr[M_{E}(V_{\omega}(j)\ldots V_{\omega}(j+n_k-1))]\,U(j+n_k)+U(j)=0.
\]
This, together with the hypothesis
\[
\big|\tr[M_{E}(V_{\omega}(j)\ldots V_{\omega}(j+n_k-1))] \big|\leq C\,,
\]
implies
\begin{eqnarray*}  
2C\, \max\{\|U(j+n_k)\|,\|U(j+2n_k)\|\}&\geq&
\|U(j+2n_k)\|+C\|U(j+n_k)\|\\ &\geq& \|U(j)\|
\end{eqnarray*} for all $1\leq j\leq l$. Then,
\begin{eqnarray*}
\|U(j+n_k)\|^2+\|U(j+2n_k)\|^2&\geq&\left({\max}\{\|U(j+n_k)\|,
\|U(j+2n_k)\|\}\right)^2 \\
&\geq& \frac{1}{4C^2}\|U(j)\|^2
\end{eqnarray*} for all $1\leq j\leq l$. Hence,
\begin{eqnarray*}
\|U\|_{l+2n_k}^2\!&=&\!\sum_{i=1}^{l+2n_k}\|U(i)\|^2\\ &=&
\sum_{i=1}^l\|U(i)\|^2+\sum_{i=l+1}^{l+2n_k}\|U(i)\|^2 \\
&\geq& \sum_{i=1}^l\|U(i)\|^2+\sum_{i=1}^l
(\|U(i+n_k)\|^2+\|U(i+2n_k)\|^2) \\
&\geq& \sum_{i=1}^l\|U(i)\|^2+\frac{1}{4C^2}
\sum_{i=1}^l\|U(i)\|^2\\ &=&\left(1+\frac{1}{4C^2}\right)
\|U\|_l^2\ .
\end{eqnarray*}
Therefore,  
\[
\|U\|_{l+2n_k}\geq D\|U\|_l,
\]
with $D=\left(1+\frac{1}{4C^2}\right)^{\frac{1}{2}}$.
\end{proof}

\begin{lemma}\label{lema13} 
Let $E\in\sigma_{\varpi}$ and~$u$ a solution to~\eqref{eqAutoval} with NIC. Then, for all $n\geq 1$, 
\[\|U\|_{2^{n+2}}\geq D\|U\|_{2^{n-1}},
\]
with $D=\left(1+\frac{1}{4C^2}\right)^{\frac{1}{2}}.$
\end{lemma}
\begin{proof}
We will apply  Lemma~\ref{lema12} to show that \[
\|U\|_{2^{n+2}}\geq D\|U\|_{2^{n-1}},
\]
for all $E\in\sigma_{\varpi}$ and NIC solution~$u$. Fix $n\geq 1$ and consider the $n$-partition of $V_{\varpi}$.

By Theorem~\ref{thmBoundTrace},  $|x_{n}|:=|\tr(M_{E}(a_n))|\leq C$, with $C>1$. Considering the structure  built of potential as in~\eqref{potencial2blocos} with the repetition of sequence in blocks, using  the  Lemma~\ref{prop2.1} and Lemma~\ref{particaohull}, we have to exhibit squares in the potential. Thus, the hypothesis~\eqref{2blocos} in Lemma~\ref{lema12}  is satisfied for $n_{k}=2^{n}$ and $l=2^{n}-1$. Consequently, 
\[
\|U\|_{2^{n+2}}\geq \|U\|_{2*2^{n}+2^{n}-1}\geq D \|U\|_{2^{n}-1}\geq D\|U\|_{2^{n-1}},
\]
with $D=\left(1+\frac{1}{4C^{2}}\right)^{\frac{1}{2}}$.
\end{proof}

By Lemma~\ref{lema13} (replace $n+2$ with $3n$), we get
\[
\|U\|_{2^{3n}}\geq D\|U\|_{2^{3n-3}}\geq \ldots \geq D^{n},
\]
and so, by picking $0<\gamma \leq \frac{ \log_{2} D}{3}$,
\[
\frac{\|U\|_{2^{3n}}}{(2^{3n})^{\gamma}} \geq \left(\frac{D}{2^{3\gamma}}\right)^{n}\geq 1.
\]
Therefore,
\begin{equation}\label{eqLBu3n}
\|u\|_{L_n}\geq C_{1}L_{n}^{\gamma},
\end{equation}
with the choices $L_{n}=2^{3n}$ and $C_{1}=1/\sqrt{2}$ (independent of~$E$).

\begin{corollary} Let $E\in\sigma_{\varpi}$ and~$u$ be a solution to~\eqref{eqAutoval} with NIC. Then, for $\gamma_1=\gamma/2$ and~ $C_{1}=1/\sqrt{2}$, one has
  \begin{equation}\label{eqboudinf}
   \|u\|_{L}\geq C_1 L^{\gamma_1}
   \end{equation} 
   for $L>0$ sufficiently large. 
   \end{corollary}
\begin{proof}
For all $n>3$ we have $4n > 3(n+1)$. Given  $L>0$, pick~$n$ so that $2^{3n}\leq L < 2^{3(n+1)}$, and by taking~\eqref{eqLBu3n} into account,
\[
\frac{1}{C_1}\|u\|_{L}\geq \frac{1}{C_1}\|u\|_{2^{3n}}\geq 2^{3n \gamma}\geq 2^{2n \gamma} = 2^{4n \frac{\gamma}{2}} > 2^{3(n+1) \frac{\gamma}{2}}> L^{ \frac{\gamma}{2}}, 
\] and \eqref{eqboudinf} follows. 
\end{proof}

\subsection{Upper bound}
Now, for each $E\in\sigma_\varpi$,  we derive the upper bound 
\begin{equation}\label{eqUpperB}
\|u\|_{L}\leq  L^{\gamma_{2}},
\end{equation}for  any solution to~\eqref{eqAutoval} with NIC. We will adapt  techniques from~\cite{IRT, IT} along with Theorem~\ref{thmBoundTrace}. There will be no restriction to values of~$L$, so that we may be able to simultaneously apply both lower and upper bounds with  the sequence $L_n$  found in Subsection~\ref{subsectLowerB}. 

Recall that, by~\eqref{eme_n}, for $n\in \N$, $M_{n+1}=M_E(a_{n+1})=M_E(a_nb_n)=M_{n-1}^2M_n$, so that 
\[
 M_{n+1}=M_{n-1}^{2}M_{n}=x_{n-1}M_{n-1}M_{n}-M_{n} 
 \]
and, by Cayley-Hamilton Theorem, 
\begin{eqnarray*}
M_{n}M_{n+1}&=&x_{n+1}M_{n}-M_{n}M_{n+1}^{-1}=x_{n+1} 
M_{n}-M_{n}(M_{n-1}^{2}M_{n})^{-1}\\
&=&x_{n+1}M_{n} + I -x_{n-1}M_{n-1}^{-1}= x_{n+1}M_{n} 
+I-x_{n-1}(x_{n-1}I-M_{n-1})\\
&=&(1-x_{n-1}^{2})I+x_{n+1}M_{n}+x_{n-1}M_{n-1}.
\end{eqnarray*}

From these recursive relations, it is natural to introduce a family of $4\times 4$ matrices whose norms will serve to estimate upper bounds for the solutions to~\eqref{eqAutoval}. For $n\in \N$, let~$B_n$ be the $4\times 4$ matrix such that 
\[\left( 
\begin{array}{c}
I \\ 
M_{n+1}\\
M_{n}\\
M_{n}M_{n+1}
\end{array}\right)=B_{n}\left( 
\begin{array}{c}
I \\ 
M_{n}\\
M_{n-1}\\
M_{n-1}M_{n}
\end{array}\right).\]
Explicitly
\[ B_{n}=\left( 
\begin{array}{cccc}
1 & 0 & 0 &0 \\ 
0 & -1 & 0 & x_{n-1}\\
0 & 1 & 0 & 0 \\
1-x_{n-1}^{2} & x_{n+1} & x_{n-1} & 0
\end{array}
\right)\ , 
\] and $\det{B_n}=x_{n-1}^2$.

For  $n,k\geq 0$, put
\[D(n,0)=I \quad  \text{and} \quad D(n,k+1)=B_{n+k+1}D(n,k).
\]
We have the following product formula, for every $n\geq 0$ and $k\in \N$,   
\[ 
D(n,k)=B_{n+k}B_{n+k-1}\cdots B_{n+1} =\prod^{1}_{i=k} B(n+i) \ ,
\]
with the (further) convention that $D(0,0)$, being an empty product, equals the $4\times4$ identity matrix~$I$. 

Denoting $Z_n=M_{n-1}M_n$, it follows by the definitions of $B_n$ and~$D(n,k)$ that 
\[\left( 
\begin{array}{c}
I \\ 
M_{n+k+1}\\
M_{n+k}\\
Z_{n+k+1}
\end{array}\right)=D(n,k)\left( 
\begin{array}{c}
I \\ 
M_{n+1}\\
M_{n}\\
Z_{n+1}
\end{array}\right) \ .\]
Let $D(n,k)_{ij}$ denote the element in $i$th row and $j$th column of~$D(n,k)$, with $i,j\in \{1,2,3,4\}$; we gather some basic relations in Lemma~\ref{lema2}.

\begin{lemma} \label{lema2} 
For each $n\geq 0$, $k\geq 0$, $k+n\in \N$, 
\[\begin{array}{ll}
D(n,k+1)_{1j}=\delta_{1j}\\
D(n,k+1)_{2j}=-1D(n,k)_{2j}+x_{n+k}D(n,k)_{4j}\\
D(n,k+1)_{3j}=D(n,k)_{2j}\\
D(n,k+1)_{4j}=(1-x_{n+k}^{2})\delta_{1j}+ 
x_{n+k+2}D(n,k)_{2j}+x_{n+k}D(n,k)_{3j} \,.\end{array}\]
\end{lemma}

\begin{lemma}\label{lema3}
If $E\in\sigma_{\varpi}$, then for all $n\geq 0$ and $k\in \N$,
\[|D(n,k)_{ij}|\leq K^{k}, \ \ \ \ i,j\in\left\{1,2,3,4\right\}\,,\]
for some constant $K>1$.
\end{lemma}
\begin{proof}
 If $E\in\sigma_{\varpi}$, by Theorem~\ref{thmBoundTrace}, there is~$n_0$ such that  
$1<\sup_{n\geq n_0}|x_n|\leq C<\infty$, for some $C>2$. Set $K:=C^{2}+2C+1$. By Lemma~\ref{lema2}, it is sufficient to prove that 
\[|D(n,k)_{2j}|\leq K^{k} \quad \text{ and } \quad |D(n,k)_{4j}|\leq K^{k} \ ,\]
which we do by induction on $k$. For $k=1$, $D(n,1)=B_{n+1}$, and so 
\begin{align*}
|D(n,1)_{2j}|& \leq |x_{n}|+1\leq C+1\leq K\,,\\
|D(n,1)_{4j}|& \leq |x_{n}|^{2} +1 +|x_{n+2}|+|x_{n}|\leq C^{2}+2C+1=K\,.
\end{align*}
The induction step reads
\begin{align*}
|D(n,k+1)_{2j}|& \leq |D(n,k)_{2j}|+|x_{n+k}||D(n,k)_{4j}|\leq (C+1)K^{k}\leq 
K^{k+1}\,,\\
|D(n,k+1)_{4j}|& \leq 
|x_{n+k}|^{2}+1+|x_{n+k+2}||D(n,k)_{2j}|+|x_{n+k}||D(n,k-1)_{2j}|\\ & \leq 
(C^{2}+2C+1)K^{k}\leq K^{k+1},
\end{align*}
and the proof is complete.
\end{proof}

\begin{lemma} 
\label{lema4} 
If $E\in\sigma_{\varpi}$ then, for each $n\geq 0$,
\[\max\left\{ \|M_{n+1}\|, \|Z_{n+1}\| \right\}\leq J^{n+1}\,,\]
where $J=\max\left\{4, K,  4\|M_{0}\|,4\|M_{1}\|,4\|Z_{1}\|\right\}$.
\end{lemma}
\begin{proof}
Note that, for each $n\geq 0$,
\[\left( 
\begin{array}{c}
I \\ 
M_{n+1}\\
M_{n}\\
Z_{n+1}
\end{array}\right)=D(0,n)\left( 
\begin{array}{c}
I \\ 
M_{1}\\
M_{0}\\
Z_{1}
\end{array}\right).\]
For $J$ as in the statement of the lemma,
\begin{eqnarray*}
\|M_{n+1}\|&=&\|D(0,n)_{21}I+ 
D(0,n)_{22}M_{1}+D(0,n)_{23}M_{0}+D(0,n)_{24}Z_{1}\|\\
&\leq&\sum_{j=1}^{4}|D(0,n)_{2j}|\max\left\{1, 
\|M_{0}\|,\|M_{1}\|,\|Z_{1}\|\right\}\\
&\leq&K^{n}4\max\left\{1, \|M_{0}\|,\|M_{1}\|,\|Z_{1}\|\right\}\leq J^{n+1}.
\end{eqnarray*}

Similarly,
\begin{eqnarray*}
\|Z_{n+1}\|&\leq&\sum_{j=1}^{4}|D(0,n)_{4j}|\max\left\{1, 
\|M_{0}\|,\|M_{1}\|,\|Z_{1}\|\right\}\\
&\leq&K^{n}4\max\left\{1, \|M_{0}\|,\|M_{1}\|,\|Z_{1}\|\right\}\leq J^{n+1}.
\end{eqnarray*}
\end{proof}

Lemma~\ref{lema4} gives the expected estimates for~$M_{n}$; recall that for this operator, with potential built along the period doubling substitution sequence, we can write any transfer matrix as products of these matrices,  i.e., $M(m)=\prod_{i=0}^{k}M_{n_{k-i}}$. 
\begin{lemma} \label{lema5} 
If $E\in\sigma_{\varpi}$, then for $n\geq 0$ and $k\in \N$, 
\[\max\left\{ \|M_{n}M_{n+k}\|,\|M_{n}Z_{n+k}\|\right\}\leq S^{n+k}\,,\]
with $S=J(4+2C)$.
\end{lemma}
\begin{proof}
We have
\begin{align*}
& \|M_{n}M_{n+k}\|=\\ 
\; &=\|M_{n}\left(D(n,k-1)_{21}I+ 
D(n,k-1)_{22}M_{n+1}+D(n,k-1)_{23}M_{n}+D(n,k-1)_{24}Z_{n+1}\right)\|\\
&=\|D(n,k-1)_{21}M_{n}+ D(n,k-1)_{22}Z_{n+1}+\\
&+D(n,k-1)_{23}(x_{n}M_{n}-I)+D(n,k-1)_{24}(x_{n}Z_{n+1}-M_{n+1})\|\\
&\leq K^{k-1}(J^{n}+J^{n+1}+CJ^{n}+1+CJ^{n+1}+J^{n+1})\\
&\leq  K^{k-1}J^{n+1}(4+2C)=K^{k-1}(J(4+2C))^{n+1}\frac{4+2C}{(4+2C)^{n+1}}\\
&\leq K^{k-1}S^{n+1}\leq S^{k+n},
\end{align*}
where $S=J(4+2C)$. 

Analogously,
\begin{align*}
& \|M_{n}Z_{n+k}\|=\\ 
&=\|M_{n}\left(D(n,k-1)_{41}I+ 
D(n,k-1)_{42}M_{n+1}+D(n,k-1)_{43}M_{n}+D(n,k-1)_{44}Z_{n+1}\right)\|\\
&=\|D(n,k-1)_{41}M_{n}+ D(n,k-1)_{42}Z_{n+1}+\\
&+D(n,k-1)_{43}(x_{n}M_{n}-I)+D(n,k-1)_{44}(x_{n}Z_{n+1}-M_{n+1})\|\\
&\leq K^{k-1}(J^{n}+J^{n+1}+CJ^{n}+1+CJ^{n+1}+J^{n+1})\\
&\leq K^{k-1}J^{n+1}(4+2C)\leq S^{k+n}.
\end{align*}
\end{proof}

\begin{lemma} \label{lema6} 
Let $E\in\sigma_{\varpi}$ and $k\geq 2$. Then, for any finite set of  positive integers $n_1<n_2<\cdots < n_k$, we have
\[\|M_{n_{1}}\ldots M_{n_{k}}\|\leq  S^{n_k+k-2}.\]
\end{lemma}
\begin{proof}
Note that 
\[\|M_{n_{1}}M_{n_{2}}\|=\|M_{n_{1}}M_{n_{1}+(n_{2}-n_{1})}\|\leq S^{n_{2}}\]
and 
\[\|M_{n_{1}}Z_{n_{2}}\|=\|M_{n_{1}}Z_{n_{1}+(n_{2}-n_{1})}\|\leq 
S^{n_{2}}\,.\]

In this proof, we write $D_{ij}=D(n_{k},n_{k+1}-n_{k}-1)_{ij}$, for $k\geq 2$. One has
\begin{eqnarray*}
\|M_{n_{1}}M_{n_{2}}M_{n_{3}}\|&=&\|M_{n_{1}}M_{n_{2}}\left(D_{21}I+ 
D_{22}M_{n_{2}+1}+D_{23}M_{n_{2}}+D_{24}Z_{n_{2}+1}\right)\|\\
&=&\|-D_{23}M_{n_{1}}- 
D_{24}M_{n_{1}}M_{n_{2}+1}+(D_{21}+x_{n_{2}}D_{23})M_{n_{1}}M_{n_{2}}+\\
&+&(D_{22}+x_{n_{2}}D_{24})M_{n_{1}}Z_{n_{2}+1}\|\\
&\leq&(4+2C)K^{n_{3}-n_{2}-1}\max\left\{\|M_{n_{1}}\|,\|M_{n_{1}}M_{n_{2}+1}\|, 
\|M_{n_{1}}M_{n_{2}}\|,\|M_{n_{1}}Z_{n_{2}+1}\|\right\}\\
&\leq&SS^{n_{3}-n_{2}-1}S^{n_{2}+1}=S^{n_{3}+1}.
\end{eqnarray*}
Similarly, 
\begin{eqnarray*} 
\|M_{n_{1}}M_{n_{2}}Z_{n_{3}}\|&=&\|M_{n_{1}}M_{n_{2}}\left(D_{41}I+ 
D_{42}M_{n_{2}+1}+D_{43}M_{n_{2}}+D_{44}Z_{n_{2}+1}\right)\|\\ 
&\leq&(4+2C)K^{n_{3}-n_{2}-1}\\ &\times& \max\left\{\|M_{n_{1}}\|,\|M_{n_{1}}M_{n_{2}+1}\|, 
\|M_{n_{1}}M_{n_{2}}\|,\|M_{n_{1}}Z_{n_{2}+1}\|\right\}\\ 
&\leq&SS^{n_{3}-n_{2}-1}S^{n_{2}+1}=S^{n_{3}+1}. 
\end{eqnarray*}

By supposing that $\|M_{n_{1}}\ldots 
M_{n_{k}}\|\leq S^{n_{k}+k-2}$ 
and $\|M_{n_{1}}\ldots Z_{n_{k}}\|\leq S^{n_{k}+k-2}$, the inductive step reads 
\begin{eqnarray*}
\|M_{n_{1}}\ldots M_{n_{k+1}}\|&=&\|M_{n_{1}}\ldots M_{n_{k}}\left(D_{21}I+ 
D_{22}M_{n_{k}+1}+D_{23}M_{n_{k}}+D_{24}Z_{n_{k}+1}\right)\|\\
&=&\|-D_{23}M_{n_{1}}\ldots M_{n_{k-1}}- D_{24}M_{n_{1}}\ldots M_{n_{k-1}}M_{n_{k}+1}+(D_{21}+\\
&+&x_{n_{k}}D_{23})M_{n_{1}}\ldots M_{n_{k}}+(D_{22}+x_{n_{k}}D_{24})M_{n_{1}}\ldots M_{n_{k-1}}Z_{n_{k}+1}\|\\
&\leq&(4+2C)K^{n_{k+1}-n_{k}-1}S^{n_{k}+1+k-2}\\
&\leq&SS^{n_{k+1}-n_{k}-1}S^{n_{k}+k-1}=S^{n_{k+1}+(k+1)-2}.
\end{eqnarray*}
\end{proof}

To finish the proof of the upper bound~\eqref{eqUpperB}, given a natural number~$m$, write it in basis~2: $m=\sum_{i=0}^{k}2^{n_{i}}$. It follows that $2^{n_k}\leq 
m$, that is, $n_k\leq \log_2m$. Therefore
\[
\|M(m)\|=\|M_{n_{1}}\ldots M_{n_{k}}\| \leq S^{n_{k}+k-2} \leq S^{2n_{k}} \leq S^{2\log_{2}m}=m^{\kappa},
\]
where $\kappa=2\log_{2}S$. 

For any solution~$u$ to~\eqref{eqAutoval} with NIC, one has,
\[
|u(m)|\leq \|M(m)\|\max\left\{|u(1)|,|u(0)|\right\}\leq m^{\kappa},
\]
since $\max\left\{|u(1)|,|u(0)|\right\}\le 1$.

Put  $\gamma_{2}=\frac{2\kappa +1}{2}$, so that $\|u\|_{L}\leq L^{\gamma_{2}}$. This, combined with~\eqref{eqboudinf}, completes the proof of Proposition~\ref{propEstimativas}.

\

\
 
 \noindent {\bf Acknowledgments.} CRdO thanks the partial support by CNPq (under contract 303503/2018-1).

\end{document}